\documentclass[11pt]{article}
\usepackage{a4wide}
\usepackage{amsmath}
\usepackage{amsfonts}
\usepackage{amssymb}

\usepackage{amsthm}
\usepackage[utf8]{inputenc}
\usepackage{cellspace}
\usepackage[small,sc]{caption}
\usepackage{booktabs, multicol, multirow}
\usepackage[paperwidth=8.5in, paperheight=11in, top=1in, bottom=1in, left=1in, right=1in]{geometry}
\usepackage{microtype}
\usepackage{enumitem}
\newtheorem{theorem}{Theorem}
\newtheorem{lemma}{Lemma}

\newtheorem{corollary}{Corollary}

\theoremstyle{definition}
\newtheorem{definition}{Definition}

\theoremstyle{remark}
\newtheorem*{remark}{Remark}

\newcommand{\F}{\mathbb{F}}

\newcommand{\colvec}[1]{\begin{bmatrix}#1\end{bmatrix}} 

\DeclareMathOperator{\poly}{poly}
\DeclareMathOperator{\id}{\mathrm{I}}
\newcommand*\conc{\mathbin{\|}}
\newcommand{\bit}[1]{\langle{#1}\rangle}

\newcommand{\kunique}{There exists a randomized data structure that takes as input positive integers $\lambda$, $n$, $\kappa$ and initializes a function $\Gamma : \bit{n}^{2^{\lambda}} \to \bit{n + \kappa + 4}^{2^{\lambda + 3}}$.
In the word RAM model with word length $w$ the data structure satisfies the following:}

\newcommand{\family}{There exists a randomized data structure that takes as input positive integers $u$, $r = u^{O(1)}$, $k$, $t$ and selects a family of functions $\mathcal{F}$ from $[u]$ to $[r]$. 
In the word RAM model with word length $w$ the data structure satisfies the following:}

\begin{document}

\title{From Independence to Expansion and Back Again\footnote{A shorter version of this paper appeares in Proceedings of STOC 2015. The results in this version slightly improves those in the proceedings version for small space, see section~\ref{sec:close-to-k-improvement}.}} 

\author{
	Tobias Christiani\\
	\small \texttt{tobc@itu.dk}\\
	\small IT University of Copenhagen
	\and
	Rasmus Pagh\\
	\small \texttt{pagh@itu.dk}\\
	\small IT University of Copenhagen
	\and
	Mikkel Thorup\\
	\small \texttt{mikkel2thorup@gmail.com}\\
	\small University of Copenhagen
}


\maketitle

\begin{abstract}
\noindent
We consider the following fundamental problems:
	\begin{itemize}
		\item Constructing $k$-independent hash functions with a space-time tradeoff close to Siegel's lower bound.
		\item Constructing representations of unbalanced expander graphs having small size and allowing fast computation of the neighbor function.
	\end{itemize}
	It is not hard to show that these problems are intimately connected in the sense that a good solution to one of them leads to a good solution to the other one.
	In this paper we exploit this connection to present efficient, recursive constructions of $k$-independent hash functions (and hence expanders with a small representation).
	While the previously most efficient construction (Thorup, FOCS 2013) 
	needed time quasipolynomial in Siegel's lower bound, our time bound is
	just a logarithmic factor from
	the lower bound.
\end{abstract}



\section{Introduction}

\emph{‘Not all those who wander are lost.’} --- Bilbo Baggins.

\medskip

The problem of designing explicit unbalanced expander graphs with near-optimal parameters is of major importance in theoretical computer science.
In this paper we consider bipartite graphs with edge set $E \subset U \times V$ where ${|U| \gg |V|}$.
Vertices in $U$ have degree $d$ and expansion is desired for subsets $S \subset U$ with $|S| \leq k$ for some parameter~$k$.
Such expanders have numerous applications (e.g.~hashing~\cite{siegel2004}, routing~\cite{broder1999}, sparse recovery~\cite{indyk2010}, membership~\cite{buhrman2002}), 
yet coming up with explicit constructions that have close to optimal parameters has proved elusive.
At the same time it is easy to show that choosing $E$ at random will give a graph with essentially optimal parameters.
This means that we can efficiently and with a low probability of error produce a description of an optimal unbalanced expander that takes space proportional to $|U|$.
Storing a complete description is excessive for most applications that, provided access to an explicit construction, would use space proportional to $|V|$.    
On the other hand, explicit constructions can be represented using constant space, 
but the current best explicit constructions have parameters $d$ and $|V|$ that are polynomial in the optimal parameters of the probabilistic constructions~\cite{guruswami2009}.
Furthermore, existing explicit constructions have primarily aimed at optimizing the parameters of the expander, with the evaluation time of the neighbor function being of secondary interest, as long as it can be bounded by $\poly \log u$.
This evaluation time is excessive in applications that, provided access to the neighbor function of an optimal expander, would use time proportional to $d$, where $d$ is typically constant or at most logarithmic in $|U|$. 

In this paper we focus on optimizing the parameters of the expander while minimizing the space usage of the representation and the evaluation time of the neighbor function.
We present randomized constructions of unbalanced expanders in the standard word RAM model.
Our constructions have near-optimal parameters, use space close to $|V|$, and support computing the $d$ neighbors of a vertex in time close to~$d$.

\paragraph{Hash functions and expander graphs} There is a close connection between $k$-independent hash functions and expanders. 
A $k$-independent function with appropriate parameters will, with some probability of failure, represent the neighbor function of a graph that expands on subsets of size $k$.
This is what we refer to as going from independence to expansion, and the fact follows from the standard union bound analysis of probabilistic constructions of expanders.
Going in the other direction, from expansion to independence, was first used by Siegel~\cite{siegel2004} as a technique for showing the existence of $k$-independent hash functions with evaluation time that does not depend on $k$.
We follow in Siegel's footsteps and a long line of work (see e.g.~\cite{dietzfelbinger2012} for an overview) that focuses on the space-time tradeoff of $k$-independent hash functions over a universe of size $u = |U|$.

Ideally, we would like to construct a data structure in the word RAM model that takes as input parameters $u$, $k$, and~$t$, and returns a $k$-independent hash function over $U$.
The hash function should use space $k(u/k)^{1/t}$ and have evaluation time $O(t)$, matching up to constant factors the space-time tradeoff of Siegel's cell probe lower bound for $k$-independent hashing~\cite{siegel2004}.
We present the first construction that comes close to matching the space-time tradeoff of the cell probe lower bound.

\paragraph{Method} Our work is inspired by Siegel's graph powering approach \cite{siegel2004} and by recent advances in tabulation hashing~\cite{thorup2013}, 
showing that it is possible to efficiently describe expanders in space much smaller than $u$.
Our main insight is that it is possible to make simple, recursive expander constructions by alternating between strong unbalanced expanders and highly random hash functions.
Similarly to previous work, we follow the procedure of letting a $k$-independent function represent a bipartite graph $\Gamma$ that expands on subsets of size $k$. 
We then apply a graph product to $\Gamma$ in order to increase the size of the universe covered by the graph while retaining expander properties.
At each step of the recursion we return to $k$-independence by combining the graph product with a table of random bits, leaving us with a new $k$-independent function that covers a larger universe.
By combining the technique of alternating between expansion and independence with a new and more efficient graph product, we can improve upon existing randomized constructions of unbalanced expanders.    

\subsection{Our contribution}
Table~\ref{tab:results} compares previous upper and lower bounds on $k$-independent hashing with our results, as presented in Corollaries 1, 2, and 3.
As can be seen, most results present a trade-off between time and space controlled by a parameter $t$.
Tight lower and upper bounds have been known only in the cell probe model,
but our new construction nearly matches the cell probe lower bound by Siegel~\cite{siegel2004}.

The time bound for the construction using explicit expanders~\cite{guruswami2009} uses the degree of the expander as a conservative lower bound, 
based on the possibility that the neighbor function in their construction can be evaluated in constant time in the word RAM model.
The time bound that follows directly from their work is $\poly \log u$. 
While the constant factors in the exponent of the space usage of~\cite{siegel2004,thorup2013} have likely not been optimized, their techniques do not seem to be able to yield space close to the cell probe lower bound. 

As can be seen our construction polynomially improves either space or time compared to each of the previously best trade-offs.
We also find our construction easier to describe and analyze than the results of~\cite{guruswami2009,kedlaya2008,thorup2013}, with simplicity comparable to that of Siegel's influential paper~\cite{siegel2004}. 

\begin{table}[t]
	
	\centering
	\renewcommand\arraystretch{1.5}
	\caption{Space-time tradeoffs for $k$-independent hash functions}
	\small
    \begin{tabular}{|c|c|c|} \hline
	{\bf Reference}  &  {\bf Space}   &  {\bf Time}  \\ \hline\hline
	Polynomials~\cite{joffe1974,wegman1981}  & $k$ & $O(k)$ \\ \hline
	Preprocessed polynomials~\cite{kedlaya2008} & $k^{1 + \varepsilon}(\log u)^{1 + o(1)}$ & $(\poly \log k) (\log u)^{1 + o(1)}$ \\ \hline
	Expanders~\cite{guruswami2009} + \cite{siegel2004} & $k^{1+ \varepsilon}d^{2}$ &  $d = O(\log(u)\log(k))^{1 + 1/\varepsilon}$ \\ \hline
    Expander powering \cite{siegel2004} & $k^{(1 - \varepsilon)t}u^{\varepsilon} + u^{1/t}$ & $O(1/\varepsilon)^{t}$ \\ \hline
    Double tabulation \cite{thorup2013} & $k^{5t} + u^{1/t}$ & $O(t)$ \\ \hline
    Recursive tabulation \cite{thorup2013} & $\poly k + u^{1/t}$ & $O(t^{\log t})$  \\ \hline 
    Corollary 1                              & $ku^{1/t}t^{3}$ & $O(t^{2} + t^{3}\log (k) / \log(u))$  \\ \hline 
    Corollary 2                              & $k^{2}u^{1/t}t^{2}$ & $O(t \log t + t^{2} \log (k) /  \log (u))$  \\ \hline
    Corollary 3*    & $ku^{1/t}t$ & $O(t \log t)$  \\ 
	\hline\hline
	    Cell probe lower bound~\cite{siegel2004} & $k(u/k)^{1/t}$            & $t < k$ \text{ probes}      \\ \hline
	    Cell probe upper bound~\cite{siegel2004}  &  $k(u/k)^{1/t}t$            & $O(t)$ \text{ probes}   \\
	\hline
    \end{tabular}

\begin{minipage}{0.83\textwidth} 
\footnotesize
\vspace{0.6em}
Table notes: Space-time tradeoffs for $k$-independent hash functions from a domain of size $u$, with the trade-off controlled by a parameter $t$. 
Time bounds in the last two rows are number of cell probes, and remaining rows refer to the word RAM model with word size $\Theta(\log u)$.
Leading constants in the space bounds are omitted.
We use $t$ to denote an arbitrary positive integer parameter that controls the trade-off, and 
We use $\varepsilon$ to denote an arbitrary positive constant.
*Corollary 3 relies on the assumption $k = u^{O(1/t)}$.
\end{minipage}
\label{tab:results}
\end{table}

Like all other randomized constructions our data structures comes with an error probability, 
but this error probability is \emph{universal} in the sense that if the construction works then it provides independent hash values on \emph{every subset} of at most $k$ elements from $U$.
This is in contrast to other known constructions~\cite{dietzfelbinger2003,pagh2008} that give independence with high probability on each \emph{particular} set of at most $k$ elements, 
but will fail almost surely if independence for a superpolynomial number of subsets is needed.

\paragraph{Applications} 
Efficient constructions of highly random functions is of fundamental interest with many applications in computer science.
A $k$-independent function can, without changing the analysis, replace a fully random function in applications that only rely on $k$-subsets of inputs mapping to random values.
We can therefore view $k$-independent functions as space and randomness efficient alternatives to fully random functions, capable of providing compact representations of complex structures such as expander graphs over very large domains.
Apart from the construction of expander graphs with a small description, as an example application, $k$-independent functions with a universal error probability can be used to construct ``real-time'' dictionaries that are able to handle extremely long (in expectation) sequences of insertion and deletion operations in constant time per operation before failing.

Let $\tau > 1$ be a constant parameter.
We use a $k$-independent hash function with $k=w^{O(\tau)}$ to split a set of $n$ machine words of $w$ bits into $O(n)$ subsets such that each subset has size at most $k$, with probability at least $1 - 2^{-w^\tau}$.
Handling each subset with Thorup's recent construction of dictionaries for sets of size $w^{O(\tau)}$ using time $O(\tau)$ per operation~\cite{thorup2014} we get a dynamic dictionary in which, with high probability, every operation in a sequence of length $\ell < 2^{O(w^\tau)}$ takes constant time.
In comparison the hash functions of~\cite{DM90,dietzfelbinger2003,pagh2008} can only guarantee that sequences of length $\ell <\text{poly}(n)$ operations, where $n < 2^w$, succeed with high probability.
The splitting hash function needs space $u^{\Omega(1)}$, which might exceed the space usage of an individual dictionary, but this can be seen as a shared resource that is used for many dictionaries (in which case we bound the total number of operations before failure).

\section{Background and overview}
In the analysis of randomized algorithms we often assume access to a fully random function of the form~${f : [u] \to [r]}$ where $[n]$ denotes the set $\{0,1, \dots, n-1\}$.
To represent such a function we need a table with $u$ entries of $\log r$ bits.
This is impractical in applications such as hashing based dictionaries where we typically have that $u \gg r$ and the goal is to use space $O(r)$ to store $r$ elements of $[u]$.
Fortunately, the analysis that establishes the performance guarantees of a randomized algorithm can often be modified to work even in the case where the function $f$ has weaker randomness properties.

One such concept of limited randomness is $k$-independence, first introduced to computer science in the 1970s through the work of Carter and Wegman on universal hashing \cite{carter1977}.  
A family of functions from $[u]$ to $[r]$ is $k$-independent if, for every subset of $[u]$ of cardinality at most $k$, 
the output of a random function from the family evaluated on the subset is independent and uniformly distributed in $[r]$.
Trivially, the family of all functions from $[u]$ to $[r]$ is $k$-independent, but representing a random function from this family uses too much space.
It was shown in \cite{joffe1974} that for every finite field $\F$ the family of functions that consist of all polynomials over $\F$ of degree at most $k-1$ is $k$-independent.
A function from this family can be represented using near-optimal space \cite{chor1985} by storing the $k$ coefficients of the polynomial.
The mapping defined by a function $f$ from a $k$-independent polynomial family over $\F = \{x_{1}, x_{2}, \dots, x_{u}\}$ takes the form 
\begin{equation}
\colvec{f(x_{1}) \\ f(x_{2}) \\ \vdots \\ f(x_{u})} = 
\colvec{ 
	x_{1}^{0} & x_{1}^{1} & \dots & x_{1}^{k-1} \\
	x_{2}^{0} & x_{2}^{1} & \dots & x_{2}^{k-1} \\
	\vdots    & \vdots    & \vdots & \vdots \\
	x_{u}^{0} & x_{u}^{1} & \dots & x_{u}^{k-1} \\
}
\colvec{a_{0} \\ a_{1} \\ \vdots \\ a_{k-1}}. \label{eq:vandermonde}
\end{equation}
The $k$-independence of the polynomial family follows from properties of the Vandermonde matrix: every subset of $k$ rows is linearly independent. 
The problem with this construction is that the Vandermonde matrix is dense, resulting in an evaluation time of $\Omega(k)$ if we simply store the coefficients of the polynomial.
The lower bounds by Siegel \cite{siegel2004}, and later Larsen \cite{larsen2012}, as presented in Table 1, show that a data structure for evaluating a polynomial of degree $k-1$ using time $t < k$ must use space at least $k(u/k)^{1/t}$.   
The data structure of~\cite{kedlaya2008} presents a step in this direction, but is still far from the lower bound for $k$-independent functions.

The quest for $k$-independent families of functions with evaluation time $t < k$ can be viewed as attempts to construct compact representations of sparse matrices that fill the same role as the Vandermonde matrix.
We are interested in compact representations that support fast computation of the sparse row associated with an element $x \in [u]$. 
An example of a sparse matrix with these properties is the adjacency matrix of a bipartite expander graph with sufficiently strong expansion properties.
For the purposes of constructing $k$-independent hash functions we are primarily interested in expanders that are highly unbalanced.

\paragraph{Expander hashing} Prior constructions of fast and highly random hash functions has followed Siegel's approach of combining expander graphs with tables of random words.
If $\Gamma$ is a $k$-unique expander graph (see Definition \ref{def:k-unique}) then we can construct a $k$-independent function by composing it with a simple tabulation function $h$.
This approach would yield optimal $k$-independent hash functions if we had access to explicit expanders with optimal parameters that could be evaluated in time proportional to the left outdegree.
Unfortunately, no explicit construction of a $k$-unique expander with optimal parameters is known.

Siegel~\cite{siegel2004} addresses this problem by storing a smaller randomly generated $k$-unique expander, say, one that covers a universe of size $u^{1/t}$.
By the $k$-independent hashing lower bound, if an expander with $|U| = u^{1/t}$ has degree $d$, then in order for it to be $k$-unique it must have a right hand side of size $|V| \geq k(u^{1/t}/k)^{1/d}$.
To give a space efficient construction of a $k$-unique expander that covers a universe of size $u$, Siegel repeatedly applies the Cartesian product to the graph.
Applying the Cartesian product $t$ times to a $k$-unique expander results in a graph that remains $k$-unique but with the left degree and size of the left and right vertex sets raised to the power $t$.
Using space $u^{1/t}$ to store an expander with degree $t$, it follows from the lower bound that the expander resulting from repeatedly applying the Cartesian product must have
\begin{equation*}
|V'| \geq (k(u^{1/t}/k)^{1/d})^{t} = k^{(1 - 1/d)t}u^{1/d}.
\end{equation*}
Setting $d = 1/\varepsilon$, the randomly generated $k$-unique expander that forms the basis of the construction has degree $O(1/\epsilon)$, leading to the expression in Table~1.
Since we need to store $|V'|$ random words in a table in order to create a $k$-independent hash function, 
Siegel's graph powering approach offers a space-time tradeoff that is far from the lower bound from our perspective where both $u$, $k$, and $t$ are parameters to the hash function.

Thorup~\cite{thorup2013} shows that, for the right choice of parameters, a simple tabulation hash function is likely to form a compact representation of a $k$-unique expander.
A simple tabulation function takes a string $x = (x_{1}, x_{2}, \dots, x_{c})$ of $c$ characters from some input alphabet $\bit{n} = \{0,1\}^{n}$, and returns a string of $d$ characters from some output alphabet $\bit{m} = \{0,1\}^{m}$.
The simple tabulation function $h : \bit{n}^{c} \to \bit{m}^{d}$ is evaluated by taking the exclusive-or of $c$ table-lookups 
\begin{equation*}
h(x) = h_{1}(x_{1}) \oplus h_{2}(x_{2}) \oplus \dots \oplus h_{c}(x_{c})
\end{equation*}
where $h_{i} : \bit{n} \to \bit{m}^{d}$ is a random function.
The advantage of a simple tabulation function compared to a fully random function is that we only need to store the random character tables $h_{1}, h_{2}, \dots, h_{c}$.
Thorup is able to show that for~${d \geq 6c}$ a simple tabulation function is $k$-unique with a low probability of failure when $k \leq (2^m)^{1/5c}$. 
Setting $n = m$ and composing the $k$-unique expander resulting from a single application of simple tabulation with another simple tabulation function, 
Thorup first constructs a hash function with space usage~$u^{1/c}$, independence~$u^{\Omega(1/c^2)}$, and evaluation time~$O(c)$. 
He then presents a second trade-off with space~$u^{1/c}$, independence~$u^{\Omega(1/c)}$, and time~$O(c^{\log c})$ that comes from applying simple tabulation recursively to the output of a simple tabulation function. 
Similar to Siegel's upper bound, the space usage of Thorup's upper bounds with respect to $k$ is much larger than the lower bound as can be seen from Table~\ref{tab:results} where the space-time tradeoff of his results have been parameterized in terms of the independence $k$.\footnote{ 
It should be noted that Thorup's analysis is not tuned to optimize the polynomial dependence on $k$, and that he gives stronger concrete parameters for some realistic parameter settings.}

\paragraph{Explicit constructions}
The literature on explicit constructions has mostly focused on optimizing the parameters of the expander, with the evaluation time of the neighbor function being of secondary interest, as long as it is bounded by $\poly \log u$.
As can be seen from Siegel's cell probe lower and upper bounds, optimal constructions of $k$-independent hash functions have evaluation time in the range $t = 1$ to $t = \log u$.
Therefore, an explicit construction, even if we had one with optimal parameters, would without further guarantees on the running time not be enough to solve our problem of constructing efficient expanders. 
Here we briefly review the construction given by Guruswami et al. \cite{guruswami2009}.
It is, to our knowledge, currently the best explicit construction of unbalanced bipartite expanders in terms of the parameters of the graph.
Their construction and its analysis is, similarly to the polynomial hash function in equation \eqref{eq:vandermonde}, algebraic in nature and inspired by techniques from coding theory, 
in particular Parvaresh-Vardy codes and related list-decoding algorithms~\cite{parvaresh2005}.
In their construction, a vertex $x$ is identified with its Reed-Solomon message polynomial over a finite field $\F$.
The $i$th neighbor of $x$ is found by taking a sequence of powers of the message polynomial over an extension field, evaluating each of the resulting polynomials in the $i$th element of $\F$, and concatenating the output.
In contrast, the constructions presented in this paper only use the subset of standard word RAM instructions that can be implemented in $AC^{0}$.
In Table 1 we have assumed that we can evaluate their neighbor function in constant time as a conservative lower bound on the performance of their construction in the word RAM model. 
Other highly unbalanced explicit constructions given in \cite{capalbo2002, tashma2007} offer a tradeoff where either one of $d$ or $|V|$ is quasipolynomial in the lower bound.
In comparison, the construction by Guruswami et al.~is polynomial in both of these parameters. 
\section{Our constructions}
In this section we present three randomized constructions of efficient expanders in the word RAM model.
Each construction offers a different tradeoff between space, time, and the probability of failure.
We present our constructions as data structures, with the randomness generated by the model during an initialization phase.
The initialization time of our data structures is always bounded by their space usage, and to simplify the exposition we therefore only state the latter.
Alternatively, our constructions could be viewed directly as randomized algorithms, taking as input a list of parameters, a random seed, and a vertex $x \in [u]$ and returning the list of neighbors of $x$.
The hashing corollaries presented in Table~1 follow directly from our three main theorems using Siegel's expander hashing technique.
\subsection{Model of computation}
The algorithms presented in this paper are analyzed in the standard word RAM model with word size $w$ as defined by Hagerup~\cite{hagerup1998}, 
modeling what can be implemented in a standard programming language like {\tt C}~\cite{kernighan1988}.
In order to show how our algorithms benefit from word-level parallelism we use~$w$ as a parameter in the analysis.
To simplify the exposition we impose the natural restriction that, for a given choice of parameters to a data structure, the word size is large enough to address the space used by the data structure.
In other words, our results are stated with $w$ as an unrestricted parameter, but are only valid when we actually have random access in constant time.

The data structures we present require access to a source of randomness in order to initialize the character tables of simple tabulation functions.
To accomodate this we augment the model with an instruction that uses constant time to generate a uniformly random and independent integer in $[r]$ where $r \leq 2^{w}$.
We note that our constructions use only the subset of arithmetic instructions required for evaluating a simple tabulation function, i.e, standard bit manipulation instructions, integer addition, and subtraction.
Our results therefore hold in a version of the word RAM model that only uses instructions that can be implemented in $AC^{0}$, known in the literature as the restricted model~\cite{hagerup1998} or the Practical RAM~\cite{miltersen1996}. 
\subsection{Notation and definitions}
Let $\bit{n} = \{0,1\}^{n}$ denote the alphabet of $n$-bit strings, and let $x = (x_{1}, x_{2}, \dots, x_{c}) \in \bit{n}^{c}$ denote a string of $n$-bit characters of length $c$.
We define a concatenation operator $\conc$ that takes as input two characters $x \in \bit{n}$ and $y \in \bit{m}$, and concatenates them to form $x \conc y \in \bit{n + m}$.
The concatenation operator can also be applied to strings of equal length where it performs component-wise concatenation.
Given strings $x \in \bit{n}^{c}$ and $y \in \bit{m}^{c}$ the concatenation $x \conc y$ is an element of $\bit{n+m}^{c}$ with the $i$th component of $x \conc y$ defined by $(x \conc y)_{i} = x_{i} \conc y_{i}$.
We also define a prefix operator. 
Given $x \in \bit{n}$ and a positive integer $m$, in the case where $m \leq n$ we use $x[m] \in \bit{m}$ to denote the $m$-bit prefix of $x$.
In the case where $m > n$ we pad the prefix such that $x[m] \in \bit{m}$ denotes $x[n] \conc 0^{m-n}$ where $0^{m-n}$ is the character consisting of a string of $m-n$ bits all set to $0$.

We will present word RAM data structures that represent functions of the form~${\Gamma : \bit{n}^{c} \to \bit{m}^{d}}$.
The function $\Gamma$ defines a $d$-regular bipartite graph with input set $\bit{n}^{c}$ and output set $\{1, 2, \dots, d\} \times \bit{m}$.
For $S \subseteq \bit{n}^{c}$ we overload $\Gamma$ and define $\Gamma(S) = \{ (i, \Gamma(x)_{i}) \mid x \in S \}$, i.e., $\Gamma(S)$ is the set of outputs of $S$.
We are interested in constructing functions where every subset $S$ of inputs of size at most $k$ contains an input that has many unique neighbors, formally:
\begin{definition} \label{def:k-unique}
Let $\Gamma : \bit{n}^{c} \to \bit{m}^{d}$ be a function satisfying the following property:   
\begin{equation*}
\forall S \subseteq \bit{n}^{c}, |S| \leq k, \exists x \in S : |\Gamma(\{x\}) \setminus \Gamma(S \setminus \{ x \})| > l.
\end{equation*}
Then, for $l \geq 0$ we say that $\Gamma$ is \emph{$k$-unique}. 
If further $l \geq d/2$ we say that $\Gamma$ is \emph{$k$-majority-unique}.
\end{definition}
For completeness we define the concept of $k$-independence:
\begin{definition} \label{def:kindependence}
Let $k$ be a positive integer and let $\mathcal{F}$ be a family of functions from $U$ to $R$.
We say that~$\mathcal{F}$ is a \mbox{\emph{$k$-independent}} family of functions if,
for every choice of~${l \leq k}$ distinct keys~${x_{1}, \dots, x_{l}}$ and arbitrary values~${y_{1}, \dots, y_{l}}$, 
then, for $f$ selected uniformly at random from $\mathcal{F}$ we have that 
\begin{equation*}
\Pr[f(x_{1}) = y_{1} \land f(x_{2}) = y_{2} \land \dots \land f(x_{k}) = y_{k}] = |R|^{-k}.
\end{equation*}
\end{definition}
Simple tabulation functions are an important tool in our constructions. 
Our data structures can be made to consist entirely of simple tabulation functions and our evaluation algorithms can be viewed as a sequence of adaptive calls to this collection of simple tabulation functions.  
\begin{definition}
Let $(R, \oplus)$ denote an abelian group.
A \emph{simple tabulation function} $h : \bit{n}^{c} \to R$ is defined by
\begin{equation*}
h(x) = \bigoplus_{i=1}^{c}h_{i}(x_{i}) 
\end{equation*}
where each \emph{character table} $h_{i} : \bit{n} \to R$ is a $k$-independent function. 
\end{definition}
In this paper we consider simple tabulation functions with character tables that operate either on bit strings under the exclusive-or operation, $R = (\bit{m}, \oplus)$, 
or on sets of non-negative integers modulo some integer $r$, $R = ([r], +)$.
\subsection{From k-uniqueness to k-independence}
In his seminal paper Siegel~\cite{siegel2004} showed how a $k$-unique function can be combined with a table of random elements in order to define a $k$-independent family of functions.
In his paper on the expansion properties of tabulation hash functions, 
Thorup~\cite[Lemma 2]{thorup2013} used a slight variation of Siegel's technique that makes use of the position-sensitive structure of the bipartite graph defined by $\Gamma : \bit{n}^{c} \to \bit{m}^{d}$.
This is the version we state here.
\begin{lemma}[Siegel {\cite{siegel2004}}, Thorup {\cite{thorup2013}}] \label{lem:expanderhashing}
Let $\Gamma : \bit{n}^{c} \to \bit{m}^{d}$ be $k$-unique and let $h : \bit{m}^{d} \to R$ be a simple tabulation function.
Then $h \circ \Gamma$ defines a family of $k$-independent functions.
We sample a function from the family by sampling the character tables of $h$.
\end{lemma}
\subsection{From k-independence to k-uniqueness}
A $k$-independent function has the same properties as a fully random function when considering $k$-subsets of inputs.
Randomized constructions of $k$-unique functions only need to consider $k$-subsets of inputs. 
We can therefore use the standard analysis of randomized constructions of bipartite expanders to show that, for the right choice of parameters, a $k$-independent function is likely to be $k$-unique.
For completeness we provide an analysis here.
In our exposition it will be convenient parameterize the $k$-uniqueness or $k$-majority-uniqueness of our constructions in terms of a positive integer~$\kappa$ such that $k = 2^{\kappa}$.
\begin{lemma} \label{lem:expanderparameters}
For every choice of positive integers $c$, $n$, $\kappa$ let~$\Gamma : \bit{n}^{c} \to \bit{m}^{d}$ be a $2^{\kappa}$-independent function.
Then,
\begin{itemize}
\item[--] for $m \geq n + \kappa + 1$ and $d \geq 4c$ we have that $\Gamma$ is $2^{\kappa}$-unique with probability at least $1 - 2^{-dn/2}$.
\item[--] for $m \geq n + \kappa + 4$ and $d \geq 8c$ we have that $\Gamma$ is $2^{\kappa}$-majority-unique with probability at least $1 - 2^{-dn/4}$.
\end{itemize}
\end{lemma}
\begin{proof}
We will give the proof for $k$-majority-uniqueness. The proof for $k$-uniqueness uses the same technique.
By a standard argument based on the pigeonhole principle, for~$\Gamma$ to be $k$-majority-unique it suffices that for all $S \subseteq \bit{n}^{c}$ with~$|S| \leq k$ we have that $|\Gamma(S)| > (3/4)d|S|$.
Given that $\Gamma$ is $k$-independent, we will now bound the probability that there exists a subset $S$ with $|S| \leq k$ such that $|\Gamma(S)| \leq (3/4)d|S|$.
For every pair of sets $(S, B)$ satisfying that $S \subseteq \bit{n}^{c}$ with $|S| \leq k$ and $B \subseteq \{1, 2, \dots, d \} \times \bit{m}$ with $|B| = (3/4)d|S|$, 
the probability that $\Gamma(S) \subseteq B$ is given by $\prod_{i = 1}^{d}(|B_{i}|/2^{m})^{|S|}$ where $B_{i} = \{ (i, y) \in B \}$.
By the inequality of the arithmetic and geometric means we have that
\begin{equation*}
\prod_{i = 1}^{d}\left(\frac{|B_{i}|}{2^{m}}\right)^{|S|} \leq \left(\frac{|B|}{d2^{m}}\right)^{d|S|}.
\end{equation*}
This allows us to ignore the structure of $B$, and obtain a union bound that matches that of the standard non-compartmentalized probabilistic construction of bipartite expanders.
The probability that $\Gamma$ fails to be $k$-majority-unique is upper bounded by
\begin{equation*}
\sum^{k}_{i=2} \binom{2^{cn}}{i} \binom{d2^{m}}{(3/4)di} \left( \frac{(3/4)di}{d2^{m}} \right)^{di}. 
\end{equation*}
For every choice of positive integers $c$, $n$, $\kappa$, for ${m \geq  n + \kappa + 4}$ and $d \geq 8 c $ we get a probability of failure less than $2^{-2cn}$.
\end{proof}
\subsection{A simple k-unique function}
In this section we introduce a simple construction of a $k$-unique function of the form $\Gamma : \bit{n}^{c} \to \bit{m}^{d}$.
We obtain $\Gamma$ as the last in a sequence $\Gamma_{1}, \Gamma_{2}, \dots, \Gamma_{c}$ of $k$-unique functions~${\Gamma_{i}: \bit{n}^{i} \to \bit{m}^{d}}$.
Each $\Gamma_{i}$ for $i>1$ is defined in terms of $\Gamma_{i-1}$.
At the bottom of the recursion we tabulate a $k$-independent function $\Gamma_{1} : \bit{n} \to \bit{m}^{d}$.
In the general step we apply $\Gamma_{i-1}$ to the length $i-1$ prefix of the key $(x_{1}, x_{2}, \dots, x_{i-1})$, concatenate the result vector component-wise with the $i$th character $x_{i}$, 
and apply a simple tabulation function $h_{i} : \bit{m + n}^{d} \to \bit{m}^{d}$.
The recursion is therefore given by
\begin{equation}
\Gamma_{i} = h_{i} \circ (\Gamma_{i-1} \conc \id^{(d)}) \label{eq:simple}
\end{equation}
where $\id^{(d)} : \bit{n} \to \bit{n}^{d}$ is the repeated identity function.
The following theorem summarizes the properties of $\Gamma$ in the word RAM model.

\begin{theorem} \label{thm:simple}
There exists a randomized data structure that takes as input positive integers $c$, $n$, $\kappa$ 
and initializes a function $\Gamma : \bit{n}^{c} \to \bit{n + \kappa + 1}^{4c}$.
In the word RAM model with word size $w$ the data structure satisfies the following:
\begin{itemize}
\item[--] The space usage is $O(2^{2n + \kappa}c^{3}(n + \kappa)/w)$.  
\item[--] The evaluation time of $\Gamma$ is $O(c^{2} + c^{3}(n + \kappa)/w)$.
\item[--] The probability that $\Gamma$ is $2^{\kappa}$-unique is at least $1 - 2^{-cn}$.
\end{itemize}
\end{theorem}
\begin{proof}
Set $m = n + \kappa + 1$ and $d = 4c$.
We initialize $\Gamma$ by tabulating a $k$-independent function $\Gamma_{1} : \bit{n} \to \bit{m}^{d}$ and simple tabulation functions $h_{2}, h_{3}, \dots, h_{c}$.
In total we need to store $c$ functions that each have $O(c)$ character tables with $O(2^{2n + \kappa})$ entries of $O(c(n + \kappa))$ bits.
The space usage is therefore $O(2^{2n + \kappa}c^{3}(n + \kappa)/w)$. 
The same bound holds for the time to initialize the data structure. 

The evaluation time of $\Gamma$ can be found by considering the recursion $\Gamma_{i} = h_{i} \circ (\Gamma_{i-1} \conc \id^{(d)})$. 
At each of the $c$ steps we perform $O(c)$ lookups and take the exclusive-or of $O(c)$ bit strings of length $O(c(n + \kappa))$.
The total evaluation time is therefore $O(c^{2} + c^{3}(n + \kappa)/w)$.

Consider the function $\Gamma_{i} = h_{i} \circ (\Gamma_{i-1} \conc \id^{(d)})$.
Conditioned on $\Gamma_{i-1}$ being $k$-unique, it is easy to see that $(\Gamma_{i-1} \conc \id^{(d)})$ is $k$-unique, and by Lemma \ref{lem:expanderhashing} we have that $\Gamma_{i}$ is $k$-independent.
For our choice of parameters, according to Lemma \ref{lem:expanderparameters} the probability that $\Gamma_{i}$ fails to be $k$-unique is less than $2^{-2cn}$.
Therefore, $\Gamma$ is $k$-unique if $\Gamma_{1}, \Gamma_{2}, \dots, \Gamma_{c}$ are $k$-unique.
This happens with probability at least $1 - c2^{-2cn} \geq 1 - 2^{-cn}$. 
\end{proof}

Combining Theorem \ref{thm:simple} and Lemma \ref{lem:expanderhashing}, we get $k$-independent hashing in the word RAM model.
We state our result in terms of a data structure that represents a family of functions~$\mathcal{F}$.
The family is defined as in Lemma \ref{lem:expanderhashing} and represented by a particular instance of a function $\Gamma$, constructed using Theorem \ref{thm:simple}, 
together with the parameters of a family of simple tabulation functions.

\begin{corollary} \label{cor:simple}
\family
\begin{itemize}
\item[--] The space used to represent $\mathcal{F}$, as well as a function $f \in \mathcal{F}$, is $O(ku^{1/t}t^{2}(\log u + t \log k)/w)$.
\item[--] The evaluation time of $f$ is $O(t^{2} + t^{2}(\log u +  t \log k)/w)$.
\item[--] With probability at least $1 - 1/u$ we have that $\mathcal{F}$ is a $k$-independent family.
\end{itemize}
\end{corollary}
\begin{proof}
We apply Theorem \ref{thm:simple}, setting $c = 2t$, ${n = \lceil (\log u)/2t \rceil}$, ${\kappa = \lceil \log k \rceil}$.
This gives a function ${\Gamma : \bit{n}^{c} \to \bit{n + \kappa + 1}^{4c}}$ that is $k$-unique over $[u]$ with probability at least $1 - 1/u$.
To sample a function from the family we follow the approach of Lemma \ref{lem:expanderhashing} and compose $\Gamma$ with a simple tabulation function $h : \bit{n + \kappa + 1}^{4c} \to [r]$.
The space used to store $\Gamma$ follows directly from Theorem \ref{thm:simple} and dominates the space used by $h$.
Similarly, the evaluation time of $h \circ \Gamma$ is dominated by the time it takes to evaluate $\Gamma$.
\end{proof}

\begin{remark}
For every integer $\tau \geq 1$ we can construct a family $\mathcal{F}^{(\tau)}$ that is $k$-independent with probability at least ${1 - u^{-\tau}}$ at the cost of increasing the space usage and evaluation time by a factor $\tau$.
The family is defined by
\begin{equation*}
\mathcal{F}^{(\tau)} = \{ f = \bigoplus_{i = 1}^{\tau} f_{i} \mid f_{i} \in \mathcal{F}_{i} \}
\end{equation*} 
where each $\mathcal{F}_{i}$ is constructed independently.
\end{remark}

\begin{remark}
The recursion in equation \eqref{eq:simple} is well suited for sequential evaluation where the task is to evaluate $\Gamma$ in an interval of $[u]$, in order to generate a $k$-independent sequence of random variables.
To see this, note that once we have evaluated $\Gamma$ on a key $x = (x_{1}, x_{2}, \dots, x_{c})$, a change in the last character only changes the last step of the recursion.   
It follows that we can generate $k$-independent variables using amortized time $O(t)$ and space close to $O(ku^{1/t})$.
To our knowledge, this presents the best space-time tradeoff for the generation of $k$-independent variables in the case where we do not have access to multiplication over a suitable finite field as in \cite{christiani2014}.  
\end{remark}
\subsection{A divide and conquer approach}
In this section we introduce a data structure for representing a $k$-majority-unique function that offers a faster evaluation time at the cost of using more space.
As in the simple construction from Theorem~\ref{thm:simple} we use the technique of alternating between expansion and independence, 
but rather than reading a single character at the time, we view the key as composed of two characters $x =(x_{1}, x_{2})$ and recurse on each. 
In the previous section we increased the size of the domain of our $k$-unique function by concatenating part of the key, forming the $k$-unique function $\Gamma \conc \id^{(d)}$.
If we use only a few large characters this approach becomes very costly in terms of the space required to store the simple tabulation function $h$ in the composition $h \circ (\Gamma \conc \id^{(d)})$.
To be able to efficiently recurse on large characters we show that the function $\Upsilon((x_{1}, x_{2})) = \Gamma(x_{1}) \conc \Gamma(x_{2})$ is $k$-unique when $\Gamma$ is $k$-majority-unique.   

\begin{lemma} \label{lem:interleaving}
Let $\Gamma : \bit{n}^{c} \to \bit{m}^{d}$ be a $k$-majority-unique function.
Then $\Gamma \conc \Gamma : \bit{n}^{c} \times  \bit{n}^{c} \to \bit{2m}^{d}$ is $k$-unique.
\end{lemma}
\begin{proof}
To ease notation we define $\Upsilon = \Gamma \conc \Gamma$.
Let $x = (x_{1}, x_{2})$ denote an element of $\bit{n}^{c} \times \bit{n}^{c}$. 
For $S \subseteq \bit{n}^{c} \times  \bit{n}^{c}$ define $S_{1,a} = \{ x \in S \mid x_{1} = a \}$.
The following holds for every $x = (x_{1}, x_{2}) \in S$.
\begin{equation} \label{eq:interleavingbound}
\begin{aligned}
|\Upsilon(\{x\}) \setminus \Upsilon(S \setminus \{x\})| &= |\Upsilon(\{x\}) \setminus (\Upsilon(S \setminus S_{1,x_{1}}) \cup \Upsilon(S_{1,x_{1}} \setminus \{x\}))| \\
&= |(\Upsilon(\{x\}) \setminus \Upsilon(S \setminus S_{1,x_{1}})) \cap (\Upsilon(\{x\}) \setminus \Upsilon(S_{1,x_{1}} \setminus \{x\}))| \\
&\geq |\Upsilon(\{x\}) \setminus \Upsilon(S \setminus S_{1,x_{1}})| + |\Upsilon(\{x\}) \setminus \Upsilon(S_{1,x_{1}} \setminus \{x\})| - |\Upsilon(\{x\})|. 
\end{aligned}
\end{equation}
We will show that for every $S \subseteq \bit{n}^{c} \times \bit{n}^{c}$ with $|S| \leq k$ there exists a key $(x_{1}, x_{2}) \in S$ such that $|\Upsilon(\{x\}) \setminus \Upsilon(S \setminus \{x\})| > 0$. 
We begin by choosing the first component of $x$.
Let $\pi_{j}(S) = \{ x_{j} \mid x \in S \}$ denote the set of $j$th components of $S$. 
By the $k$-majority-uniqueness of $\Gamma$, considering the set $\pi_{1}(S)$, we have that
\begin{equation*}
\exists x_{1} \in \pi_{1}(S) : \forall x \in S_{1,x_{1}} : |\Upsilon(\{x\}) \setminus \Upsilon(S \setminus S_{1,x_{1}})| > d/2.
\end{equation*}
Fix $x_{1}$ with this property and consider the choice of $x_{2}$.
By the $k$-majority-uniqueness of $\Gamma$, considering the set $\pi_{2}(S)$, we have that
\begin{equation*}
 \forall x_{1} \in \pi_{1}(S) : \exists x_{2} \in \pi_{2}(S_{1,x_{1}}) : |\Upsilon(\{x\}) \setminus \Upsilon(S_{1,x_{1}} \setminus \{x\})| > d/2.
\end{equation*}
We can therefore always find a key $(x_{1}, x_{2}) \in S$ such that both $|\Upsilon(\{x\}) \setminus \Upsilon(S \setminus S_{1,x_{1}})| > d/2$, $|\Upsilon(\{x\}) \setminus \Upsilon(S_{1,x_{1}} \setminus \{x\})| > d/2$ are satisfied.
The result follows from equation \eqref{eq:interleavingbound} where we use the fact that $|\Upsilon(\{x\})| = d$.
\end{proof}

We will give a recursive construction of a $k$-majority-unique function of the form $\Gamma_{i} : \bit{n}^{2^{i}} \to \bit{m}^{2^{i + 3}}$.
Let $h_{i} : \bit{2m}^{2^{i +2}} \to \bit{m}^{2^{i+3}}$ be a simple tabulation function.
For $i > 0$ the recursion takes the following form.  
\begin{equation}
\Gamma_{i} = h_{i} \circ (\Gamma_{i-1} \conc \Gamma_{i-1}). \label{eq:majorityconcatenation}
\end{equation}
At the bottom of the recursion we tabulate a $k$-independent function $\Gamma_{0}$.

\begin{theorem} \label{thm:majorityrecursion}
\kunique
\begin{itemize}
\item[--] The space usage is $O(2^{2(n + \kappa + \lambda)}(n + \kappa)/w)$. 
\item[--] The evaluation time of $\Gamma$ is $O(2^{\lambda}(\lambda + 2^{\lambda}(n + \kappa)/w))$.
\item[--] With probability at least $1 - 2^{-2n + 1}$ we have that $\Gamma$ is $2^{\kappa}$-majority-unique.
\end{itemize}
\end{theorem}
\begin{proof}
Let $m = n + \kappa + 4$.
We initialize $\Gamma$ by tabulating $\Gamma_{0}$ and the character tables of the simple tabulation functions $h_{1}, h_{2}, \dots, h_{\lambda}$ where $h_{i} : \bit{2m}^{2^{i +2}} \to \bit{m}^{2^{i+3}}$.
In total we have $O(2^{\lambda})$ tables with $O(2^{2(n + \kappa)})$ entries of $O(2^{\lambda}(n + \kappa))$ bits, resulting in a total space usage of $O(2^{2(n + \kappa + \lambda)}(n + \kappa)/w)$.

Let $T(i)$ denote the evaluation time of $\Gamma_{i}$.
For $i = 0$ we can evaluate $\Gamma_{0}$ by performing a single lookup in $O(1)$ time.
For $i > 0$ evaluating $h_{i} \circ (\Gamma_{i-1} \conc \Gamma_{i-1})$ takes two evalutions of $\Gamma_{i-1}$ followed by evaluating $h_{i}$ on their concatenated output using $O(2^{i}(1 + 2^{i}(n + \kappa)/w))$ operations.
The recurrence takes the form
\begin{equation*}
T(i) \leq 
\begin{cases} 2T(i-1) + O(2^{i}(1 + 2^{i}(n + \kappa)/w)) & \mbox{if } i > 0 \\ 
O(1) & \mbox{if } i = 0 
\end{cases}
\end{equation*}
The solution to the recurrence is $O(2^{i}(i + 2^{i}(n + \kappa)/w))$.

We now turn our attention to the probability that $\Gamma_{i} = h_{i} \circ (\Gamma_{i-1} \conc \Gamma_{i-1})$ fails to be $k$-majority-unique.
Conditional on $\Gamma_{i-1}$ being $k$-majority-unique, by Lemma \ref{lem:interleaving} we have that $(\Gamma_{i-1} \conc \Gamma_{i-1})$ is $k$-unique and composing it with $h_{i}$ gives us a $k$-independent function.
For our choice of parameters, according to Lemma \ref{lem:expanderparameters} the probability that $\Gamma_{i}$ fails to be $k$-majority-unique is less than $2^{-2^{i + 1}n}$.
Therefore, $\Gamma$ is $k$-majority-unique if $\Gamma_{0}, \Gamma_{1}, \dots, \Gamma_{\lambda}$ are $k$-majority-unique.
This happens with probability at least $1 - \sum_{i=0}^{\lambda} 2^{-2^{i + 1}n} \geq 1 - 2^{-2n + 1}$. 
\end{proof}

\begin{remark}
The recursion in equation \eqref{eq:majorityconcatenation} is well suited for parallelization. 
If we have $c$ processors working in lock-step with some small shared memory we can evaluate $\Gamma$ with domain $\bit{n}^{c}$ in time~$O(c)$. 
\end{remark}

\begin{corollary}
\family
\begin{itemize}
\item[--] The space used to represent $\mathcal{F}$, as well as a function $f \in \mathcal{F}$, is $O(ku^{1/t}t(\log u + t\log k)/w)$.
\item[--] The evaluation time of $f$ is $O(t \log t + t (\log u + t \log k)/w)$.  
\item[--] With probability at least $1 - u^{-1/t}$ we have that $\mathcal{F}$ is a $k$-independent family.
\end{itemize}
\end{corollary}
\begin{proof}
Apply Theorem \ref{thm:majorityrecursion} with parameters $\lambda = \lceil \log t \rceil + 1$, $n = \lceil (\log u) / 2t \rceil + 1$, and $\kappa = \lceil \log k \rceil$.
This gives is a function $\Gamma$ that is $k$-unique over $[u]$ with probability at least $1 - u^{-1/t}$.
The family $\mathcal{F}$ is defined by the composition of $\Gamma$ with a suitable simple tabulation function following the approach of Lemma \ref{lem:expanderhashing}.
\end{proof}
\subsection{Balancing time and space}
Theorem \ref{thm:simple} yielded a $k$-unique function over $\bit{n}^{c}$ with an evaluation time of about $O(c^{2})$ while using linear space in $k$.
Theorem \ref{thm:majorityrecursion} resulted in an evaluation time of about~$O(c \log c)$, using quadratic space in $k$.
Under a mild restriction on $k$, the two techniques can be combined to obtain an evaluation time of $O(c \log c)$ and linear space in $k$.
We take the construction from Theorem \ref{thm:majorityrecursion} as our starting point, 
but instead of tabulating the character tables of $h_{1}, \dots, h_{\lambda}$ we replace them with more space efficient $k$-independent functions that we construct using Theorem \ref{thm:simple}.

\begin{theorem} \label{thm:linear}
There exists a randomized data structure that takes as input positive integers $\lambda$, $n$, $\kappa = O(n)$ 
and initializes a function $\Gamma : \bit{n}^{2^{\lambda}} \to \bit{n + \kappa + 4}^{2^{\lambda + 3}}$.
In the word RAM model with word size $w$ the data structure satisfies the following:
\begin{itemize}
\item[--] The space usage is $O(2^{n + \kappa + 2\lambda}n/w)$.  
\item[--] The evaluation time of $\Gamma$ is $O(2^{\lambda}(\lambda + 2^{\lambda}n/w))$.
\item[--] With probability at least $1 - 2^{-n + 1}$ we have that $\Gamma$ is $2^{\kappa}$-majority-unique.
\end{itemize}
\end{theorem}
\begin{proof}
At the top level, the recursion underlying $\Gamma$ takes the same form as in Theorem \ref{thm:majorityrecursion}.
\begin{equation*}
\Gamma_{i} = h_{i} \circ (\Gamma_{i-1} \conc \Gamma_{i-1}).
\end{equation*}
The functions $h_{i} : \bit{2m}^{2^{i+2}} \to \bit{m}^{2^{i+3}}$ are simple tabulation functions with $m = n + \kappa + 4$.
Each $h_{i}$ is constructed from~$2^{i+2}$ character tables $h_{i,j} : \bit{2m} \to \bit{m}^{2^{i+3}}$.
Theorem~\ref{thm:majorityrecursion} only assumes that the character tables $h_{i,j}$ are $k$-independent functions.
We will apply Theorem~\ref{thm:simple} to construct a function~$\Upsilon$ that we for each character table $h_{i,j}$ compose with a simple tabulation function $g_{i,j}$ in order to construct $h_{i,j}$.
By the restriction that $\kappa = O(n)$ we have that $m = O(n)$.
We set the parameters of~$\Upsilon$ to $\hat{c} = O(1)$, $\hat{n} = \lceil n/2 \rceil$, $\hat{\kappa} = \kappa$ such that $\bit{2m}$ can be embedded in $\bit{\hat{n}}^{\hat{c}}$.
Furthermore,~$\Upsilon$ uses~$O(2^{n + \kappa}n/w)$ words of space, can be evaluated in $O(1)$ operations, and is $k$-unique with probability at last ${1 - 2^{-n-1}}$.
Because $\Upsilon$ has $O(1)$ output characters, the time to evaluate~$h_{i,j} = g_{i,j} \circ \Upsilon$ is no more than a constant times the word length of the output of $h_{i,j}$.
The time to evaluate $\Gamma$ therefore only increases by a constant factor compared to the evaluation time in Theorem~\ref{thm:majorityrecursion}.

The probability of failure of $\Gamma$ to be $k$-majority-unique is the same as in Theorem \ref{thm:majorityrecursion}, provided that $\Upsilon$ does not fail to be $k$-unique.
This gives a total probability of failure of less than $2^{-2n + 1} + 2^{-n -1} < 2^{-n + 1}$.

We only store a single $\Upsilon$ and the character tables of $g_{i,j}$ that we use to simulate the character tables $h_{i,j}$.
From the parameters of $\Upsilon$ we have that $g_{i,j}$ uses $O(1)$ character tables with $O(2^{n+\kappa})$ entries of $O(2^{i}n/w)$ words.
The space usage is dominated by the $O(2^{\lambda})$ character tables of $h_{\lambda}$ that use space $O(2^{n + \kappa + 2\lambda}n/w)$ in total.
\end{proof}

\begin{corollary}
There exists a randomized data structure that takes as input positive integers $u$, $r = u^{O(1)}$, $t$, $k = u^{O(1/t)}$ and selects a family of functions $\mathcal{F}$ from $[u]$ to $[r]$. 
In the word RAM model with word length $w$ the data structure satisfies the following:
\begin{itemize}
\item[--] The space used to represent $\mathcal{F}$, as well as a function $f \in \mathcal{F}$, is $O(ku^{1/t}t(\log u)/w)$.
\item[--] The evaluation time of $f$ is $O(t \log t + t(\log u)/w)$.
\item[--] With probability at least $1 - u^{-1/t}$ we have that $\mathcal{F}$ is a $k$-independent family.
\end{itemize}
\end{corollary}
\begin{proof}
Apply Theorem \ref{thm:linear} with parameters $\lambda = \lceil \log t \rceil$, $n = \lceil (\log u) / t \rceil + 1$, and $\kappa = \lceil \log k \rceil$.
This gives is a function $\Gamma$ that is $k$-unique over $[u]$ with probability at least $1 - u^{-1/t}$.
The family $\mathcal{F}$ is defined by the composition of $\Gamma$ with a suitable simple tabulation function following the approach of Lemma \ref{lem:expanderhashing}.
\end{proof}
\subsection{An improvement for space close to $k$}\label{sec:close-to-k-improvement}
In this section we present a different space efficient version of the divide-and-conquer recursion.
The new recursion is based on an extension of the ideas behind the graph product from Lemma \ref{lem:interleaving}.
In Lemma \ref{lem:interleaving} we use expansion properties over subsets of size $k$ and concatenate the output characters of $\Gamma$, resulting in an output domain of size at least $k^{2}$.
By using stronger expansion properties and modifying our graph concatenation product to fit the structure of the key set, we are able to reduce the space usage at the cost of using more time.
We now introduce a property that follows from stronger edge expansion.

\begin{definition} \label{def:k-super-majority-unique}
Let $\Gamma : \bit{n}^{c} \to \bit{m}^{d}$ be a function satisfying the following property:   
\begin{equation*}
\forall S \subseteq \bit{n}^{c}, |S| \leq k, \exists A \subseteq S, |A| > |S|/2 : \forall x \in A : |\Gamma(\{x\}) \setminus \Gamma(S \setminus \{ x \})| > d/2.
\end{equation*}
Then we say that $\Gamma$ is \emph{$k$-super-majority-unique}.
\end{definition}

The following lemma shows how we can construct a $k$-unique function over $U^{2}$ from a set of $k$-super-majority-unique functions over $U$.
For a bit string $x$ we will use the notation $x[m]$ to denote the $m$-bit prefix of $x 0^m$, i.e., a zero-padded $m$-bit prefix of $x$.

\begin{lemma} \label{lem:superconcatenation}
Let $q$ be a positive integer. 
For $j = 1, 2, \dots, q$ let $\Gamma_{j} : \bit{\kappa}^{c} \to \bit{m_{j}}^{d}$ be $\min(2k^{j/q}, k)$-super-majority-unique and set $m = \max_{j}(m_{j} + m_{q-j+1})$.
Then the function $\Gamma : \bit{\kappa}^{c} \times \bit{\kappa}^{c} \to \bit{m}^{dq}$ defined by
\begin{equation}
\Gamma(x_{1}, x_{2})_{(j-1)q + l} = (\Gamma_{j}(x_{1})_{l} \conc \Gamma_{q-j+1}(x_{2})_{l})[m] \textnormal{ for } (j,l) \in \{1, \dots, q\} \times \{1,\dots, d\}
\end{equation}
is $k$-unique.
\end{lemma}
\begin{proof}
Consider a set of keys $S \subseteq \bit{n}^{c} \times \bit{n}^{c}$ with $|S| \leq k$.
We will show that there exists an index $j \in \{ 1, \dots, q \}$ and a key $x = (x_{1},x_{2}) \in S$ 
such that $x$ has a unique neighbor with respect to $S$ and $\Gamma_{j} \conc \Gamma_{q-i+1}$.
Consider the set of first components of the set of keys $\pi_{1}(S)$. 
For some $j \in \{1, \dots, q\}$ we must have that $k^{(j-1)/q} \leq |\pi_{1}(S)| \leq k^{j/q}$.
By the super-majority-uniqueness properties of $\Gamma_{j}$ there must exist more than $k^{(j-1)/q}/2$ first components $x_{1} \in \pi_{1}(S)$ such that $\Gamma_{j}(x_{1})$ has more than $d/2$ unique neighbors with respect to $\pi_{1}(S)$.
Furthermore, because $|S| \leq k$, there exists at least one such $x_{1}$ that is a component of at most $\min(2k^{(q-j+1)/q}, k)$ keys.
Following a similar argument to the proof of Lemma \ref{lem:interleaving}, by the majority-uniqueness properties of $\Gamma_{q-j+1}$ there exists $x_{2} \in S_{1,x_{1}}$ such that we get a unique neighbor. 
\end{proof}
In the following lemma we use a single $k$-independent function to represent a set of $k$-super-majority-unique functions such that the concatenated product of these functions is $k$-unique.
The proof of the lemma is omitted since it follows from using the approach of Lemma \ref{lem:expanderparameters} to obtain expansion $|\Gamma(S)| > (7/8)d|S|$, 
and then applying Lemma \ref{lem:superconcatenation} to obtain the $k$-uniqueness property. 
\begin{lemma} \label{lem:superexpanderparameters}
For every choice of positive integers $c$, $q$, $\kappa$, let~$f : \bit{\kappa}^{c} \to \bit{2\kappa + 12}^{16cq}$ be a~${2^{\kappa}}$-independent function.
For $j = 1, \dots, q$ define $\Gamma_{j} : \bit{\kappa}^{c} \to \bit{\lceil ((j + 1)/q)\kappa \rceil + 12}^{16cq}$ by
\begin{equation}
\Gamma_{j}(x)_{l} = f(x)_{l}[\lceil ((j + 1)/q)\kappa \rceil + 12] \textnormal{ for } l \in \{1, \dots, 16cq \}.
\end{equation}
Let $m = \lceil (1 + 3/q)\kappa \rceil + 26$. 
Then the function $\Gamma : \bit{\kappa}^{c} \times \bit{\kappa}^{c} \to \bit{m}^{16cq^{2}}$ defined by
\begin{equation}
\Gamma(x_{1}, x_{2})_{(j-1)q + l} =
(\Gamma_{j}(x_{1})_{l} \conc \Gamma_{q-j+1}(x_{2})_{l})[m] \textnormal{ for } (j,l) \in \{1, \dots, q\} \times \{1,\dots, 16cq\}
\end{equation}
is $k$-unique with probability at least $1 - 2^{-2c\kappa}$ .
\end{lemma}
\noindent We remind the reader that the notation $x[m]$ is used to denote the zero-padded $m$-bit prefix of $x$.
Taking the prefix of the concatenated output characters of $\Gamma_{j}$ and $\Gamma_{q-j+1}$ is done with the sole purpose of padding the output characters of $\Gamma$ to uniform length.

We now define a randomized recursive construction of a $k$-unique function similar to the one in Theorem \ref{thm:majorityrecursion}.
The parameters of the data structure are $\lambda$, $\kappa$, and $q$.
The parameters $\lambda$ and $\kappa$ determine the size of the universe and the desired $k$-uniqueness.
The parameter $q$ controls the space-time tradeoff of the character tables used in the recursion.
At the outer level of the recursion, for $i = 1, \dots, \lambda$, we repeatedly square the size of the domain, 
constructing $k$-unique functions of the form $\Gamma_{i} : \bit{\kappa}^{2^{i}} \to \bit{2\kappa + 26}^{144 \cdot 2^{i}}$.	
At level $i$ of the recursion, we obtain a $k$-independent function by composing $\Gamma_{i}$ with a simple tabulation function $h_{i+1} : \bit{2\kappa + 26}^{144 \cdot 2^{i}} \to \bit{2 \kappa + 12}^{48 \cdot 2^{i+1}}$.
The output of this function is then used to construct $\Gamma_{i+1}$, following the approach of Lemma \ref{lem:superexpanderparameters} with the parameter $q$ set to $3$.
For $i = 1, 2, \dots, \lambda$ the recursion is described by the following set of equations
\begin{equation}
\begin{aligned}
\Gamma_{i}(x_{1}, x_{2})_{(j-1)48 \cdot 2^{i} + l} &= (\Gamma_{i,j}(x_{1})_{l} \conc \Gamma_{i,4-j}(x_{2})_{l})[2\kappa + 26]\\
\Gamma_{i,j}(x_{s})_{l} &= h_{i}(\Gamma_{i-1}(x_{s}))_{l}[\lceil ((j + 1)/3)\kappa \rceil + 12]\\
\Gamma_{0}(x_{s})_{l} &= \id^{(48)}(x_{s})_{l}[2\kappa + 26]
\end{aligned}
\end{equation}
where the indices are $j \in \{1,2,3\}$, $l \in \{1, \dots, 48 \cdot 2^{i} \}$, and $s \in \{1, 2\}$.
We have defined $\Gamma_{0}$ by simply repeating the input $48$ times, padded to length $2\kappa + 26$, to ensure that it fits into the recursion. 
In practice we only require $h_{1} \circ \Gamma_{0}$ be be $k$-independent over domain $\bit{\kappa}$.  

To further reduce the space usage we apply the technique from Lemma \ref{lem:superexpanderparameters} to implement the character tables of $h_{i}$.
Each character table has domain $\bit{2\kappa + 26}$. 
We view this domain as consisting of two characters of length $\kappa' = \kappa + 13$.
We apply Lemma \ref{lem:superexpanderparameters} with parameters $c = 1$, $q$, and $\kappa = \kappa'$ to construct a function 
$\Upsilon : \bit{\kappa'}^{2} \to \bit{\lceil (1 + 3/q)\kappa' \rceil + 26}^{16q^{2}}$ that is $k$-unique with probability at least $1 - 2^{-2\kappa'}$. 
To facilitate fast evaluation we tabulate the $k$-independent function $f' : \bit{\kappa'} \to \bit{\lceil (1 + 1/q)\kappa' \rceil + 12}^{16q}$ used to construct $\Upsilon$.
The $j$th character table of $h_{i}$ is constructed by composing $\Upsilon$ with an appropriate simple tabulation function,
\begin{equation}
h_{i,j} = \Upsilon \circ g_{i,j},
\end{equation}
where $g_{i,j} : \bit{\lceil (1 + 1/q)\kappa' \rceil + 12}^{16q} \to \bit{2\kappa + 12}^{48 \cdot 2^{i}}$ is tabulated.

\begin{theorem} \label{thm:prefixrecursion}
There exists a randomized data structure that takes as input positive integers $\lambda$, $\kappa$, $q$, and initializes a function $\Gamma : \bit{\kappa}^{2^{\lambda}} \to \bit{2\kappa + 26}^{48 \cdot 2^{\lambda}}$.
In the word RAM model with word length $w$ the data structure satisifes the following:
\begin{itemize}
\item[--] The space usage is $O(2^{(1 + 3/q)\kappa + 2\lambda} q^{2} \kappa /w)$. 
\item[--] The evaluation time of $\Gamma$ is $O(2^{\lambda} q^{2} (\lambda + 2^{\lambda}\kappa/w))$.
\item[--] With probability at least $1 - 2^{-2(\kappa - 1)}$ we have that $\Gamma$ is $2^{\kappa}$-unique.
\end{itemize}
\end{theorem}
\begin{proof}
The total space usage is dominated by the simple tabulation functions used to implement the character tables of $h_{\lambda}$.
There are $O(2^{\lambda})$ simple tabulation functions $g_{i,j}$. 
Each of these has $O(q^{2})$ character tables with a domain of size $O(2^{(1 + 3/q)\kappa})$ that map to bit strings of length $O(2^{\lambda}\kappa)$.
This gives a total space usage of $O(2^{(1 + 3/q)\kappa + 2\lambda)} q^{2} \kappa /w)$.

Let $T(i)$ denote the evaluation time of $\Gamma_{i}$.
For $i = 1$ we can evaluate $\Gamma_{1}$ by performing a constant number of lookups into $h_{0}$ and combine prefixes of the output in $O(1)$ time.
For $i > 1$ evaluating $\Gamma_{i}$ takes two evaluations of $\Gamma_{i-1}$ and an additional amount of work combining prefixes that is only a constant factor greater than the time required to read the output of $h_{i} \circ \Gamma_{i-1}$.
Evaluating $h_{i}$ is performed by $O(2^{i})$ evaluations of character tables of the form $g_{i,j} \circ \Upsilon$.
The degree of $\Upsilon$ is $O(q^{2})$ and it has an evaluation time that is proportional to the degree.
We therefore perform $O(q^{2})$ lookups into the character tables of $g_{i,j}$ where we read bit strings of length $O(2^{i}\kappa)$.
The recurrence describing the evaluation time of $\Gamma_{i}$ takes the form
\begin{equation*}
T(i) \leq 
\begin{cases} 2T(i-1) + O(2^{i}q^{2}(1 + 2^{i}\kappa/w)) & \mbox{if } i > 1 \\ 
O(1) & \mbox{if } i = 1. 
\end{cases}
\end{equation*}
The solution to the recurrence is $O(2^{i}q^{2}(i + 2^{i}(n + \kappa)/w))$.

The construction fails if $\Upsilon$ fails to be $k$-unique or if $\Gamma_{1}, \dots, \Gamma_{\lambda}$ fails to be $k$-unique. 
According to Lemma \ref{lem:superexpanderparameters} this happens with probability less than $2^{-2\kappa'} + \sum_{i=1}^{\lambda}2^{-2^{i}\kappa} < 2^{-2(\kappa - 1)}$   
\end{proof}

\begin{corollary}\label{cor:prefixcorollary}
There exists a randomized data structure that takes as input positive integers $u$, $r = u^{O(1)}$, $t$, $k$ and selects a family of functions $\mathcal{F}$ from $[u]$ to $[r]$. 
In the word RAM model with word length $w$ the data structure satisfies the following:
\begin{itemize}
\item[--] The space used to represent $\mathcal{F}$, as well as a function $f \in \mathcal{F}$, is $O(ku^{1/t}t^{2}\log(k)/w)$.
\item[--] The evaluation time of $f$ is $O(t^{2} (\log(k)/\log u)(\log(\log(u)/\log k) + \log(u)/w))$.
\item[--] With probability at least $1 - k^{-2}$ we have that $\mathcal{F}$ is a $k$-independent family.
\end{itemize}
\end{corollary}
\begin{proof}
Assume without loss of generality that $k \leq u$ and apply Theorem \ref{thm:prefixrecursion} with parameters 
$\lambda = \lceil \log(\log(u)/\log k) \rceil + 1$, $\kappa = \lceil \log k \rceil + 1$, and $q = \lceil 3t \log(k)/\log u \rceil$.
This gives is a function $\Gamma$ that is $k$-unique over $[u]$ with probability at least $1 - k^{-2}$.
We compose $\Gamma$ with a suitable simple tabulation function $h$ that maps to elements of $[r]$.
Implementing $h$ using $\Upsilon$ we get the same bounds on the space usage, evaluation time, and probability of failure as for the data structure used to represent $\Gamma$.
\end{proof}

\begin{remark}
The construction in Corollary \ref{cor:prefixcorollary} presents an improvement in the case where we wish to minimize the space usage. 
For $w = \Theta(\log u)$ and $t = \lceil \log u \rceil$ we get a space usage of $O(k \log(u) \log( k))$ and an evaluation time of $O(\log (u)\log (k)\log (\log (u) /\log (k)))$.
In comparison, for these parameters Corollary \ref{cor:simple} gives a space usage of $O(k \log^{2} u)$ and an evalution time of $O(\log^{2}(u)\log(k))$. 
\end{remark}

\section{Conclusion}
We have presented new constructions of $k$-independent hash functions that come close to Siegel's lower bound on the space-time tradeoff for such functions.
An interesting open problem is whether the gap to the lower bound can be closed.
From the perspective of efficient expanders it would be very interesting to achieve space $o(k)$ while preserving computational efficiency.
Of course, such a result is not possible via $k$-independence.

\section{Acknowledgements}
The research of Tobias Christiani and Rasmus Pagh has received funding from the European Research Council under the European Union's Seventh Framework Programme (FP7/2007-2013) / ERC grant agreement no.~614331.

Mikkel Thorup's research is partly supported by Advanced Grant DFF-0602-02499B from the Danish Council for Independent Research under the Sapere Aude research career programme.

We thank the STOC reviewers for insightful comments that helped us improve the exposition.
\bibliography{focs}
\bibliographystyle{plain}
\end{document}